\documentclass[11pt,letterpaper]{article}

\usepackage[letterpaper, left=1in, right=1in, top=1in, bottom=1in]{geometry}
\usepackage[american]{babel}
\usepackage[normalem]{ulem}
\usepackage{amsmath, amssymb, cases, amsthm}
\usepackage{thmtools}
\usepackage[shortlabels]{enumitem}
\usepackage{mdframed}
\usepackage{bbm}
\usepackage{bm}
\usepackage{microtype}
\usepackage{xcolor}
\usepackage{makecell}
\usepackage{mathtools}
\usepackage{algorithmic}
\usepackage[procnumbered,ruled,vlined,linesnumbered]{algorithm2e}
\usepackage{float}
\usepackage{varwidth}
\usepackage{modletter}
\usepackage{tcolorbox}
\usepackage{tikz}
\usepackage{cancel}

\usetikzlibrary{arrows.meta,positioning,fit,calc}

\newtcolorbox{construction}[2][]
{
	colframe = gray!50,
	colback  = gray!10,
	coltitle = gray!10!black,
	left*=0mm,
	before skip = 10pt,
	after skip = 10pt,
	title    = \textbf{\space\space #2},
	#1,
}

\SetKwInput{KwData}{Input}
\SetKwInput{KwResult}{Output}
\SetKwInput{KwGlobalVar}{Global variables}
\SetKw{Continue}{continue}
\SetKwComment{Comment}{$\triangleright$\ }{}

\usepackage{xcolor}
\definecolor{ForestGreen}{rgb}{0.1333,0.5451,0.1333}
\definecolor{DarkRed}{rgb}{0.80,0,0}
\definecolor{Red}{rgb}{1,0,0}
\usepackage[linktocpage=true,
	pagebackref=true,colorlinks,
	linkcolor=DarkRed,citecolor=ForestGreen,
	bookmarks,bookmarksopen,bookmarksnumbered]
{hyperref}

\usepackage[capitalize,nosort,nameinlink]{cleveref}

\declaretheorem[numberwithin=section]{theorem}
\declaretheorem[numberlike=theorem]{lemma}

\declaretheorem[numberlike=theorem,style=definition]{definition}
\declaretheorem[numberlike=theorem,style=remark]{remark}

\def\final{0}  %
\ifnum\final=0  %
	\newcommand{\todo}[1]{{\color{red}[{\tiny TODO: \bf #1}]\marginpar{\color{red}*}}}
	\newcommand{\znote}[1]{{\color{blue}[{\tiny zeh \bf #1}]\marginpar{\color{blue}*}}}
	\newcommand{\thatchaphol}[1]{{\color{purple}[{\tiny Thatchaphol: \bf #1}]\marginpar{\color{purple}*}}}
	\newcommand{\ky}[1]{{\color{orange}[{\tiny ky: \bf #1}]\marginpar{\color{orange}*}}}
\else %
	\newcommand{\znote}[1]{}
	\newcommand{\thatchaphol}[1]{}
	\newcommand{\todo}[1]{}
	\newcommand{\ky}[1]{}
\fi

\newcommand{\srch}{\textsc{Search}}
\newcommand{\rsrch}{\textsc{ResumeSearch}}
\newcommand{\expl}{\textsc{Explore}}
\newcommand{\augm}{\textsc{Augment}}

\title{Incremental Directed Minimum Cut\\ by Dynamizing Gabow's Algorithm}

\author{Thatchaphol Saranurak\thanks{University of Michigan, United States. Email: \texttt{thsa@umich.edu}. Supported by NSF Grant CCF-2238138 and a Sloan Fellowship.}
\and
Kaiyang Xie\thanks{University of Science and Technology of China, China. Email: \texttt{xkaiy219@gmail.com}}
\and
Zhaienhe Zhou\thanks{University of Science and Technology of China, China. Email: \texttt{zhaienhezhou@gmail.com}}}

\date{}

\begin{document}
\sloppy

\maketitle
\pagenumbering{gobble}
\begin{abstract}
We give the first incremental algorithm for directed global minimum cut. Given a directed graph with $n$ vertices undergoing $m$ edge insertions, our deterministic algorithm explicitly maintains a global minimum cut or certifies that its value is at least $k$ in $O(km\log n)$ total update time. Prior work required either that $k\le2$ or that the graph is undirected.

Our algorithm is a strict incremental extension of Gabow's state-of-the-art static algorithm~\cite{gabow1991matroid}, with no asymptotic loss in running time over the entire insertion sequence.
\end{abstract}

\pagenumbering{arabic}
\section{Introduction}

Maintaining connectivity under updates is a central problem in dynamic graph algorithms. In this paper, we study the \emph{incremental directed global minimum-cut} problem: given a directed graph undergoing edge insertions, maintain a global minimum cut, or certify that its value is at least $k$.

Formally, an $n$-vertex $m$-edge directed graph $G=(V,E)$ is \emph{$k$-edge-connected} if it remains strongly connected after deleting any set of fewer than $k$ edges. A minimum edge set whose deletion destroys strong connectivity is called a \emph{global minimum cut}. Thus, maintaining a global minimum cut is the dynamic counterpart of testing directed $k$-edge-connectivity. As usual, we assume $m\ge n$ throughout.

In the static setting, Gabow~\cite{gabow1991matroid} gave an $O(km\log n)$-time algorithm for testing directed $k$-edge-connectivity. For small $k$, this remains the fastest algorithm.
Despite many recent developments
\cite{chekuri2021faster,cen2022minimum,quanrud2025approximating,mosenzon2025almost,jiang2025approximating}, an exact algorithm with almost-linear time independent of $k$ remains a major open problem.

What happens in dynamic graphs? In \emph{undirected} graphs, dynamic edge-connectivity and minimum cut are very well developed~\cite{thorup2007fully,goranci2018incremental,goranci2023fully,jin2024fully,el2025fully,de2025tree,kenneth2026simple}. For example, one can maintain $k$-edge-connectivity in $n^{o(1)}$ update time even for $k=n^{o(1)}$ in fully dynamic graphs~\cite{el2026deterministic}.

In directed graphs, however, much less is known. Fast fully dynamic algorithms are unlikely even for the special case of strong connectivity, i.e.\ $k=1$: under standard conjectures, such algorithms require $n^{1-o(1)}$ update time~\cite{henzinger2015unifying,probst2020new}. It is therefore natural to focus on incremental and decremental settings.

Even in these restricted settings, prior nontrivial results for directed graphs were known only for very small connectivity. For $k=1$, there are near-linear total update time algorithms for incremental and decremental strong connectivity~\cite{italiano1986amortized,bernstein2019decremental,van2024almost}. For $k=2$, Georgiadis et al.\ gave incremental and decremental algorithms with super-linear total time for maintaining 2-edge-connectivity information in directed graphs~\cite{georgiadis2018incremental,georgiadis2017decremental,georgiadis2025faster}. Before this work, no nontrivial dynamic algorithm was known for directed global minimum cut once $k\ge 3$.

Our main result is the first such algorithm.

\begin{theorem}\label{thm:main}
	There is a deterministic algorithm that, given a directed graph $G$ with $n$ vertices undergoing $m$ edge insertions, explicitly maintains a global minimum cut $C\subseteq E(G)$ or reports that its size is at least $k$ using $O(km\log n)$ total update time.

	When the cut $C$ is maintained, the algorithm also explicitly maintains a vertex set $S$ and indicates that $C=E(S,V\setminus S)$ or $C=E(V\setminus S,S)$.
\end{theorem}

For every polylogarithmic $k$, this gives near-linear total update time. In particular, for $k\ge 3$, it is the first nontrivial dynamic global minimum-cut algorithm in directed graphs.

A central feature of our result is that it is a \emph{strict extension} of Gabow's classical static algorithm~\cite{gabow1991matroid} to the incremental setting \emph{without any slowdown}  over the entire insertion sequence. Moreover, the algorithm does not merely detect when the edge-connectivity crosses the threshold $k$; it explicitly maintains a global minimum cut throughout the update sequence until the cut size grows to $k$.

This extension is \emph{not} immediate. In the static setting, once Gabow's algorithm detects a minimum cut of size less than $k$, it terminates. In the incremental setting, however, future edge insertions may destroy the minimality of the current cut, so the algorithm must continue and find a new minimum cut if one exists. Restarting Gabow's algorithm from scratch after each such event would be too slow.

Our conceptual contribution is to show that, instead of restarting, one can resume the previous run of Gabow's algorithm and continue it efficiently after further insertions.

\paragraph{Expository contribution.}
In addition to dynamizing Gabow’s algorithm \cite{gabow1991matroid}, we present the static algorithm in a self-contained form. We hope that this exposition, which separates the round-robin framework, the structure of component-local augmenting paths, the warm-up brute-force implementation, and the cyclic-scanning implementation, will also be useful to readers who wish to understand Gabow’s algorithm independently of the incremental extension.

\paragraph*{Related work for fixed-pair connectivity.}
For a fixed ordered pair $(s,t)$, recent dynamic algorithms for \emph{thresholded} min-cost flow imply $m^{1+o(1)}$ total update time algorithms for testing whether $s$ and $t$ remain $k$-edge-connected in incremental and decremental directed graphs~\cite{chen2024almost,van2024almost}. This also yields algorithms that keep track of the edge-connectivity between $s$ and $t$ up to value $k$ in total time $km^{1+o(1)}$.\footnote{See also~\cite{gupta2018simple} for a simple $O(mk)$-time incremental algorithm.}

These algorithms rely on interior-point-method machinery, crucially exploiting the fact that the min-cost flow admits an $O(m)$-size linear program. In contrast, for global directed minimum cut, the best-known linear programming formulations have size $O(mn)$. Thus, this recent framework does not seem to extend to our setting.

\section{Overview}

Fix a directed graph $G=(V,E)$ and a root $a\in V$. For $X\subseteq V$, let
\(
\rho_G(X):=E(V\setminus X,X)
\)
denote the \emph{in-cut} of $X$. For an edge set $A$ and a vertex $v$, let
\(
\deg_A^-(v):=|A\cap \rho_G(v)|
\)
denote the indegree of $v$ with respect to $A$.

\paragraph*{From global minimum cut to rooted minimum cut.}
Fix an arbitrary root $a\in V$. A \emph{rooted cut} is an edge set whose removal makes some vertex unreachable from $a$. To find a global minimum cut, it suffices to solve the following rooted problem on both $G$ and the reverse graph $G^{\mathrm{rev}}$ and return the smaller cut: maintain, for a fixed root $a$, a minimum rooted cut $\rho_G(X)$ with $X\subseteq V\setminus\{a\}$, or certify that every such cut has size at least $k$.

\paragraph*{Edmonds' packing view of rooted connectivity.}
When we say that an edge set is a forest or a spanning tree, we ignore edge directions. A \emph{rooted $k$-intersection} is an edge set that can be partitioned into $k$ forests $T=T_1\sqcup\dots\sqcup T_k$ such that $\deg_T^-(a)=0$ and $\deg_T^-(v)\le k$ for every $v\neq a$. It is \emph{complete} if $\deg_T^-(v)=k$ for all $v\neq a$; in that case each $T_i$ is a spanning tree, though not necessarily an arborescence.

Edmonds~\cite{edmonds1973edge} showed the duality between rooted minimum cuts and $k$-intersections.
\begin{lemma}[Edmonds' Characterization]
	\label{lem:edmonds}
	The graph $G$ contains a complete rooted $k$-intersection if and only if every nonempty $X\subseteq V\setminus\{a\}$ satisfies
	\(
	|\rho_G(X)|\ge k.
	\)
\end{lemma}
This motivates the following algorithmic strategy: try to build a complete rooted $k$-intersection. If we succeed, then we conclude that the rooted minimum cut size is now at least $k$. If we fail, this is because a minimum rooted cut has size less than $k$, and we will need to return it.

\paragraph*{The state in stage $k$.}
The algorithm proceeds in stages $k=1,2,\dots$. At the beginning of stage $k$, we already have a complete $(k-1)$-intersection and try to extend it to a complete $k$-intersection.

During this stage, we maintain a \emph{realization} $T=(T_1,\dots,T_k)$ such that
\begin{enumerate}
	\item $T_1\cup\dots\cup T_k$ is a rooted $k$-intersection, and
	\item $T_1\cup\dots\cup T_{k-1}$ is a complete rooted $(k-1)$-intersection.
\end{enumerate}
Thus every vertex $v\neq a$ has indegree either $k-1$ or $k$ in $T:=T_1\cup\dots\cup T_k$. We call $v$ \emph{deficient} if its indegree is $k-1$.
Observe that in this stage, any rooted cut of size less than $k$ has size exactly $k-1$ and is minimum.

The connected components of $T_k$ are the key objects. We call such a component an \emph{$T_k$-component}.
We will ensure that every $T_k$-component not containing $a$ has exactly one deficient vertex.

The goal of stage $k$ is to keep merging $T_k$-components until $T_k$ itself becomes a spanning tree.
Once we succeed, we successfully build a complete rooted $k$-intersection. If we fail, we must find a rooted cut of size less than $k$.

\paragraph*{Augmenting paths and locality.}
To merge a $T_k$-component $K$, the algorithm starts from its deficient vertex $s$ and searches for an \emph{augmenting path}, formally defined in \Cref{sec:prelim}. Such a path is a sequence of edge exchanges with two key effects: (1) $s$ is no longer deficient, and (2) it merges $K$ with another $T_k$-component.

A crucial fact from Gabow's work is that one can always choose a \emph{component-local augmenting path}: if the path starts in $K$, then every edge on the path except the last one has both endpoints in $K$, and only the final edge leaves $K$ to join another $T_k$-component. %

\paragraph*{Locality yields parallel augmentation.}

The locality of component-local augmenting paths allows us to search from many $T_k$-components in parallel. The algorithm resembles Bor\r{u}vka's MST algorithm.

In one round, all current $T_k$-components are marked active. For each active $T_k$-component $K$, we search from its deficient vertex. If the search finds a component-local augmenting path whose last edge joins $K$ to another $T_k$-component $K'$, then both $K$ and $K'$ are marked inactive for the rest of the round. After all searches finish, we augment along all found paths simultaneously.

This yields geometric progress: each successful search deactivates at most two active $T_k$-components, so at least half of the active $T_k$-components succeed in every round. Hence the number of $T_k$-components drops by at least a factor of two per round, and there are only $O(\log n)$ rounds in one stage.

\paragraph*{Dynamizing Gabow's.}
Up to this point, this is essentially Gabow's static algorithm. The only real issue is what to do with a failed search.

In the static setting, a failed search will reveal a rooted cut of size less than $k$ and the algorithm may stop. In the incremental setting, however, the failure may be temporary: a later edge insertion can create a new augmenting path. Restarting the search from scratch after each insertion would lose the running time bound.

Our main idea is to \emph{pause} a failed search rather than discard it. The search is a reachability process in the auxiliary graph (whose definition we omit here), and an edge insertion in $G$ only adds new vertices and arcs there. Thus previously reached states remain reached. We store the explored region of each failed search and, after future insertions, resume only from newly reachable states. Once a resumed search reaches a joining edge, we rerun the search within the relevant $T_k$-component to extract a component-local augmenting path, then continue the process.

\paragraph*{Why the running time matches Gabow's.}
To implement resumable searches efficiently, we use Gabow's cyclic-scanning procedure. For a fixed stage $k$, this lets us charge the total work of all search and resume operations in one round to the graph edges, for a total of $O(m)$ time in that round.

Combining this with the $O(\log n)$ rounds per stage and the outer loop over $k$ stages, we obtain $O(km\log n)$ total update time for the rooted problem. Running the same algorithm on both $G$ and $G^{\mathrm{rev}}$ gives the same bound, up to a factor of two, for directed global minimum cut.

\section{Preliminaries}
\label{sec:prelim}

We use the terminology from the overview, including realizations and deficient vertices. For an edge set $A\subseteq E$, let $U(A)$ denote its underlying undirected graph.

\paragraph*{Joining edges and fundamental paths.}
Let $F$ be a forest. An $F$-component is a connected component of $U(F)$.
An edge $e=(u,v)\in E$ is \emph{joining for $F$} if $u$ and $v$ lie in different $F$-components.
If $u$ and $v$ are in the same $F$-component, let $P_F(u,v)$ denote the unique path in $U(F)$. If $e=(u,v)$ is not joining for $F$, write $P_F(e):=P_F(u,v)$ and call it the \emph{fundamental path of $e$ with respect to $F$}.

\paragraph*{Auxiliary graph and augmenting paths.}
Given a realization $T=(T_1,\dots,T_k)$, the \emph{auxiliary graph} $D(T)$ has vertex set
\[
	V\sqcup E\sqcup\{t\},
\]
where $t$ is a distinguished sink. We call the members of the copy of $E$ in $V(D(T))$ \emph{edge-vertices}, and write $A(D(T))$ for the arc set. Throughout, \emph{edges} belong to the original graph $G$, whereas directed connections in $D(T)$ are \emph{arcs}. The arcs are of the following four types.
\begin{enumerate}
	\item \textbf{Deficit-source arcs:} \\For every deficient vertex $v$ and every edge $e\in \rho_G(v)\setminus T$, add an arc $(v,e)$.
	\item \textbf{Joining-sink arcs:}\\ For every $T_k$-joining edge $e$, add an arc $(e,t)$.
	\item \textbf{Same-head exchange arcs:}\\ For every $e\in T$ and $f\in E\setminus T$ with $\mathrm{head}(e)=\mathrm{head}(f)$, add an arc $(e,f)$.
	\item \textbf{Tree exchange arcs:}\\ For every $i\in [k]$, every $f\notin T_i$, and every $e\in P_{T_i}(f)$, add an arc $(f,e)$.
\end{enumerate}

\begin{definition}[Augmenting path]
	\label{def:aug_path}
	An \emph{augmenting path} is a directed path in $D(T)$ that starts at a deficient vertex $s\in V$ and ends at the sink $t$. Equivalently, it is a sequence
	\[
		P=(s,e_1,e_2,\dots,e_r,t)
	\]
	where each $e_j\in E$ and consecutive elements are connected by arcs of $D(T)$.
\end{definition}

\paragraph*{Augmentation.}
Next, we describe how an application of an augmenting path changes each forest $T_i$.
Let $\mathrm{idx}:E\to [k]\cup\{\emptyset\}$ be the index function, where
$\mathrm{idx}(e)=i$ if $e\in T_i$ and $\mathrm{idx}(e)=\emptyset$ if $e\notin T$.

Given an augmenting path $P=(s,e_1,\ldots,e_r,t)$, write
$\alpha_j=\mathrm{idx}(e_j)$ as the indices before augmentation.
Augmentation along $P$ changes the indices of $e_1,\ldots,e_r$ from
\[
	(\alpha_1,\alpha_2,\ldots,\alpha_r)
\]
to
\[
	(\alpha_2,\alpha_3,\ldots,\alpha_r,k).
\]
Equivalently, for each $j=1,\ldots,r-1$ with
$\alpha_{j+1}\neq\emptyset$, edge $e_j$ replaces $e_{j+1}$ in
$T_{\alpha_{j+1}}$, i.e.
\[
	T_{\alpha_{j+1}} \gets T_{\alpha_{j+1}}-e_{j+1}+e_j,
\]
and then $e_r$ is inserted into $T_k$.
Thus, $e_1$ enters the solution,
each edge $e_j$ takes over the old label of its successor, and the last edge $e_r$ merges a $T_k$-component of $s$ with another $T_k$-component.  This increases $\deg_T^-(s)$ by one and leaves every other indegree unchanged.
See \Cref{fig:augmentation-exchanges}.

Our algorithm description in \Cref{sec:incremental} usually omits the sink vertex $t$ in the augmenting path and stops at a $T_k$-joining edge $e_r$. This is purely for convenience.

\begin{figure}[t]
	\centering
	\begin{tikzpicture}[
			>=Latex,
			font=\small,
			edge/.style={draw, rounded corners=2pt, minimum width=8mm, minimum height=6mm, inner sep=1.5pt, fill=white},
			idx/.style={draw, rounded corners=2pt, fill=gray!10, minimum width=8mm, minimum height=5mm, inner sep=1pt, font=\scriptsize},
			lab/.style={font=\scriptsize},
			expl/.style={font=\scriptsize, align=center, text width=13.0cm},
			aux/.style={->, thick},
			move/.style={->, thick, densely dashed},
			exchange/.style={->, thick},
			op/.style={draw, rounded corners=3pt, fill=gray!5, inner sep=3pt, align=center, font=\scriptsize},
			comp/.style={draw, dashed, rounded corners=4pt, inner sep=5pt},
			v/.style={circle, draw, inner sep=1pt, minimum size=5mm, font=\scriptsize}
		]

		\node[edge] (s)  at (-1.35,0) {$s$};
		\node[edge] (e1) at (0,0) {$e_1$};
		\node[edge] (e2) at (1.55,0) {$e_2$};
		\node[edge] (e3) at (3.10,0) {$e_3$};
		\node        (dots) at (4.55,0) {$\cdots$};
		\node[edge] (erm) at (6.00,0) {$e_{r-1}$};
		\node[edge] (er)  at (7.65,0) {$e_r$};
		\node[edge] (t)   at (9.00,0) {$t$};

		\draw[aux] (s) -- (e1);
		\draw[aux] (e1) -- (e2);
		\draw[aux] (e2) -- (e3);
		\draw[aux] (e3) -- (dots);
		\draw[aux] (dots) -- (erm);
		\draw[aux] (erm) -- (er);
		\draw[aux] (er) -- node[above=5pt, midway, lab, fill=white, inner sep=1pt] {joins $T_k$} (t);

		\node[lab, anchor=east] at (-1.70,0.90) {$\mathrm{idx}_{\rm old}$};
		\node[idx] (b1) at (0,0.90) {$\alpha_1=\emptyset$};
		\node[idx] (b2) at (1.55,0.90) {$\alpha_2$};
		\node[idx] (b3) at (3.10,0.90) {$\alpha_3$};
		\node[idx] (bm) at (6.00,0.90) {$\alpha_{r-1}$};
		\node[idx] (br) at (7.65,0.90) {$\alpha_r$};

		\node[lab, anchor=east] at (-1.70,-0.95) {$\mathrm{idx}_{\rm new}$};
		\node[idx] (a1) at (0,-0.95) {$\alpha_2$};
		\node[idx] (a2) at (1.55,-0.95) {$\alpha_3$};
		\node[idx] (a3) at (3.10,-0.95) {$\alpha_4$};
		\node[idx] (am) at (6.00,-0.95) {$\alpha_r$};
		\node[idx] (ar) at (7.65,-0.95) {$k$};

		\draw[move] (b2.south) to[out=-90,in=90] (a1.north);
		\draw[move] (b3.south) to[out=-90,in=90] (a2.north);
		\draw[move] (br.south) to[out=-90,in=90] (am.north);

		\node[lab, anchor=east] at (-1.70,-2.25) {forest exchanges};

		\node[op] (op1) at (0.78,-2.25)
		{$T_{\alpha_2}$: delete $e_2$, add $e_1$};
		\node[op] (op2) at (2.35,-3.10)
		{$T_{\alpha_3}$: delete $e_3$, add $e_2$};
		\node at (4.45,-2.70) {$\cdots$};
		\node[op] (opm) at (6.83,-2.25)
		{$T_{\alpha_r}$: delete $e_r$, add $e_{r-1}$};
		\node[op] (opp) at (8.25,-3.10)
		{$T_k$: add $e_r$};

		\draw[exchange] (e2.south) to[out=-100,in=80] (op1.north);
		\draw[exchange] (e1.south) to[out=-80,in=100] (op1.north);
		\draw[exchange] (e3.south) to[out=-100,in=80] (op2.north);
		\draw[exchange] (e2.south) to[out=-80,in=100] (op2.north);
		\draw[exchange] (er.south) to[out=-100,in=80] (opm.north);
		\draw[exchange] (erm.south) to[out=-80,in=100] (opm.north);
		\draw[exchange] (er.south) to[out=-85,in=90] (opp.north);

		\node[expl] at (4.05,-3.85)
		{Equivalently, for $j=1,\ldots,r-1$, edge $e_j$ replaces $e_{j+1}$ in the forest that used to own $e_{j+1}$; if $\alpha_{j+1}=\emptyset$, that exchange box is omitted. The last edge $e_r$ is then inserted into $T_k$.};

		\begin{scope}[yshift=-5.80cm]
			\node[lab, anchor=east] at (-1.70,-0.95) {component effect};

			\node[v] (sv) at (0,-0.8) {$s$};
			\node[v] (xv) at (1.0,-0.35) {};
			\node[v] (uv) at (2.0,-0.8) {$u$};
			\node[v] (zv) at (1.0,-1.35) {};
			\node[v] (qv) at (-0.85,-0.35) {};
			\draw[thick] (sv) -- (xv) -- (uv) -- (zv) -- (sv);
			\draw[thick] (qv) -- (sv);
			\node[comp, fit=(sv)(xv)(uv)(zv)(qv), label=above:{}] (Ki) {};
			\node[lab] at (0.60,0.25) {$K_i$: component containing $s$};

			\node[v] (yv) at (5.45,-0.8) {$v$};
			\node[v] (rv) at (6.50,-0.35) {$r_j$};
			\node[v] (wv) at (6.50,-1.35) {};
			\draw[thick] (yv) -- (rv) -- (wv) -- (yv);
			\node[comp, fit=(yv)(rv)(wv), label=above:{}] (Kj) {};
			\node[lab] at (5.95,0.25) {$K_j$: another $T_k$-component};

			\draw[aux] (qv) -- node[above, lab] {$e_1$ enters} (sv);

			\draw[aux, very thick] (yv) -- node[above, lab] {$e_r$} node[below, lab] {insert into $T_k$} (uv);

			\node[lab, align=center] at (0.75,-2.15)
			{$e_1$ and all non-final exchanges\\stay inside $K_i$};
			\node[lab, align=center] at (5.95,-2.15)
			{final step: $e_r$ joins $K_j$ to $K_i$\\and merges the two $T_k$-components};
			\node[lab, align=center] at (3.25,-2.95)
			{$\deg_T^-(s)$ increases by one; every other vertex loses and gains one incoming edge along the exchange chain.};
		\end{scope}

	\end{tikzpicture}
	\caption{Augmentation along $P=(s,e_1,\ldots,e_r,t)$. The index vector changes from $(\alpha_1,\ldots,\alpha_r)$ to $(\alpha_2,\ldots,\alpha_r,k)$.}
	\label{fig:augmentation-exchanges}
\end{figure}

\paragraph*{A shortcut-free tree-exchange lemma.}
After augmentation, each forest $T_i$ undergoes a sequence of exchanges. The lemma below gives a sufficient condition for these exchanges to preserve acyclicity.

\begin{lemma}[Shortcut-Free Tree-Exchange Lemma]
	\label{lem:matroid_exchange}
	Let $F$ be a forest, and let $(e_1,f_1),\dots,(e_t,f_t)$ be pairs such that $e_i\notin F$, $f_i\in F$, and $f_i\in P_F(e_i)$ for every $i$. Assume moreover that for every $i<j$, $f_j\notin P_F(e_i)$.
	Then $F+\{e_1,\dots,e_t\}-\{f_1,\dots,f_t\}$ is again a forest.
\end{lemma}
\begin{proof}
	For each $j$, removing $f_j$ from $F$ separates its connected component into two parts. Since $f_j\in P_F(e_j)$, the edge $e_j$ crosses this cut.

	Suppose the final graph contains a cycle $C$, and choose the largest index $j$ such that $e_j\in C$. Since $e_j$ crosses the cut of $F-f_j$, the cycle $C$ must contain another edge crossing the same cut. No edge of $F-f_j$ crosses it. If this other edge is some $e_i$ with $i<j$, then $f_j\in P_F(e_i)$, contradicting the assumption. If it is some $e_i$ with $i>j$, this contradicts the choice of $j$. Hence no such cycle exists, and the final graph is a forest.
\end{proof}
Intuitively, the condition says that no later exchange deletes an edge from the fundamental path that justified an earlier exchange. Later, this will be exactly the treewise shortcut-free property of a component-local augmenting path.

\section{Incremental Extension of Gabow's Algorithm}
\label{sec:incremental}

We prove \Cref{thm:main} in this section.
We begin by restating the static ingredients of Gabow's algorithm in the language of this paper. These are the round-robin framework, the structure of augmenting paths, and the fact that the augmenting paths found in one round can be applied simultaneously. We isolate these facts because our incremental algorithm reuses them almost verbatim.

The only new issue in the incremental setting is the following. In the static algorithm, an unsuccessful search certifies that the current graph already has a rooted cut of size less than $k$, so the algorithm may terminate (by \cref{lem:edmonds}). In the incremental setting, such a failed search is only a certificate for the \emph{current} graph: after future edge insertions, a new augmenting path may appear. Our main contribution is to show that Gabow's search can be \emph{paused} when it fails and later \emph{resumed} after edge insertions, while preserving the same structural invariants.

The section is organized as follows. In \Cref{sec:rr_framework} we state the round-robin framework. In \Cref{sec:static_gabow} we record the static structural properties inherited from Gabow. In \Cref{sec:search_interface} we describe the new incremental issue and state the guarantees required from the search procedures $\srch,\rsrch$. In \Cref{sec:brute_augment} we give a brute-force pause/resume implementation. In \Cref{sec:cyclic_scanning} we give the linear-time cyclic-scanning implementation. Finally, in \Cref{sec:proof_main_incremental} we complete the proof of \Cref{thm:main}.

\subsection{Round-Robin Framework}
\label{sec:rr_framework}

Fix a value of $k$, and suppose that we already have a complete $(k-1)$-intersection. As in Gabow's algorithm, the goal is to extend it to a complete $k$-intersection. We maintain a realization
\[
	T=(T_1,\ldots,T_k),
\]
where $T_1,\ldots,T_{k-1}$ are spanning trees and $T_k$ is a forest. Every connected component of $T_k$ (except the one containing $a$) has exactly one deficient vertex, and we call each such component a \emph{$T_k$-component}.

The algorithm proceeds in rounds. In one round, every active $T_k$-component not containing $a$ attempts to find an augmenting path from its deficient vertex. Whenever an augmenting path joins two $T_k$-components, those two components are marked inactive for the rest of the round. At the end of the round, all augmenting paths found in that round are applied simultaneously.

Compared with the static algorithm, the only additional feature is how we handle an unsuccessful search. In the static setting, failure is terminal. In the incremental setting, we instead pause that search and resume it after future edge insertions.

\begin{algorithm}[H]
	\caption{Round Robin for Incremental $k$-intersection}\label{alg:round_robin}
	$k\gets 0$\;
	\While{true}{
		$k\gets k+1$\;
		\While{$T_k$ is not a spanning tree}{\label{line:while_non_full}
			Let $K \coloneqq \{K_1, K_2, \cdots, K_q\}$ be the connected components of $U(T_k)$\;
			Mark all $K_i$ as active\;
			\For{each $i \in [q]$}{
				\If{$K_i$ does not contain $a$, and is still active}{\label{line:check_active}
					Let $s$ be the deficient vertex in $K_i$\;
					\If{$\srch(s) = \bot$}{
						\Repeat{$\rsrch(s,e) \neq \bot$}{
						Report $k-1$\tcp*{$\rho_G(\{s\}\cup V(\text{explored edge-vertices of }D(T)))$ is a cut of size $k-1$}\label{line:report_k_minus_1}
							Receive new edge insertion $e=(x,y)$\;
						}
					}

					$P\gets$ Reset all exploration states inside $K_i$ and re-run $\srch(s)$\;\label{line:search_again}
					Assuming the last joining edge in $P$ joins $K_i$ and $K_j$, we mark $K_i$ and $K_j$ as inactive\;
				}
			}
			For every path $P$ produced during this round, $\augm(P)$\;\label{line:augment_all}
		}
	}
\end{algorithm}

One might consider simplifying the algorithm by extracting the path directly from $\rsrch$. However, an augmenting path returned by $\rsrch$ does not naturally satisfy the desired treewise shortcut-free property~\cref{def:good_aug_path} (intuitively, a resumed BFS does not necessarily preserve the shortest-path distances) in the presence of edge insertions. To avoid complicating the proof, we reset the exploration states and re-run a static search. Since each component triggers this re-run at most once per round, this re-run does not affect the overall asymptotic complexity.

We remark that the pseudocode is written in an offline style for simplicity. In an online implementation, one stores the loop state and the exploration state as global variables and resumes the algorithm after each edge insertion.

The framework above is useful only if the search procedures have the right structural properties. We next recall some facts from Gabow's static algorithm.

\subsection{Static Ingredients Inherited from Gabow}
\label{sec:static_gabow}

This subsection is purely expository: it restates the static ingredients from Gabow in our notation. No incremental issue appears yet.

\subsubsection*{Component-local augmenting paths}

The basic definition of augmenting path from \Cref{def:aug_path} is not strong enough for our round-robin algorithm, because we must later augment along many paths simultaneously. We therefore use the following stronger notion.

\begin{definition}[Component-Local Augmenting Path]
	\label{def:good_aug_path}
	Let $P=(s,e_1,\ldots,e_\ell,t)$ be an augmenting path in a realization
	$T=(T_1,\ldots,T_k)$, where $T_1,\ldots,T_{k-1}$ are spanning trees rooted
	at $a$. We call $P$ a \emph{component-local augmenting path} if it satisfies the following
	properties:
	\begin{itemize}
		\item \textbf{Treewise Shortcut-Free:}
		      For every tree $T_i$ and every tree exchange arc $(e_j,e_{j+1})$ of $P$ with
		      $e_{j+1}\in P_{T_i}(e_j)$,
		      \[
			      \{e_{j+2},\ldots,e_\ell\}\cap P_{T_i}(e_j)=\emptyset.
		      \]

		\item \textbf{No Premature Joining Edge:} $\forall j<\ell, e_j$ is not $T_k$-joining. Moreover, for every tree exchange arc $(e_j,e_{j+1})$ of $P$ with $j+1<\ell$ and $e_{j+1}\in P_{T_i}(e_j)$, $P_{T_i}(e_j)$ contains no $T_k$-joining edge.

		\item \textbf{Deepest Final Joining Edge:}
		      If the last exchange of $P$ is $(e,f)$, where $f\in T_i$ is
		      $T_k$-joining, then $f$ maximizes $d(f)$ among all
		      $T_k$-joining edges in $P_{T_i}(e)$, breaking ties by the smallest
		      edge identifier. Here, $d(g)$ is the depth of an edge
		      $g\in T_i$, for $i<k$, in $U(T_i)$ rooted at $a$.
	\end{itemize}
\end{definition}

Note that any shortest augmenting path satisfies both \emph{Treewise Shortcut-Free} and \emph{No Premature Joining Edge} properties, but our algorithm will find a component-local augmenting path even though it is not the shortest path.
The intuition of \emph{Deepest Final Joining Edge} condition is as follows: assume $e \coloneqq (u,v)$ where $u,v \in V(K_i)$, the condition enforces that $f$ is the first edge leaving $V(K_i)$ from either side of the fundamental path $P_{T_i}(e)$.

The conditions in \Cref{def:good_aug_path} are without loss of generality. Indeed, starting from any augmenting path, one can repeatedly remove shortcuts, truncate at the earliest joining edge whenever one appears, and choose the deepest possible final joining edge, ensuring that all three properties are satisfied.

\cref{lem:augment_path_structure} justifies the naming of \emph{component-local} augmenting path. It shows that a component-local augmenting path stays in a single $T_k$-component until its last step, and that augmenting along it merges exactly two $T_k$-components while preserving validity. This is a fact from the static algorithm by Gabow.

\begin{restatable}[Component Locality of Augmenting Paths~\cite{GABOW1995259}]{lemma}{AugmentingPathStructure}
	\label{lem:augment_path_structure}

	Let $P$ be a component-local augmenting path starting from a deficient vertex $s$ in a $T_k$-component $K_i$. Then, every edge in $P$ except for the final $T_k$-joining edge has both endpoints in $V(K_i)$. The final $T_k$-joining edge connects a vertex in $V(K_i)$ to a vertex in $V(K_j)$ for some $j \neq i$.

	Moreover, augmenting along $P$ results in a valid $k$-intersection where $T_1, \dots, T_{k-1}$ are spanning trees, and:
	\begin{itemize}
		\item Apart from the final exchange possibly removing the joining edge $e_r$ from its previous tree, augmenting along $P$ modifies only the internal structure of $T[V(K_i)]$. Adding $e_r$ to $T_k$ then merges $K_i$ and $K_j$.

		\item The augmentation increases $\deg_T^-(s)$ by one but leaves the indegrees of all other vertices unchanged. Consequently, the merged component has the deficient vertex of $K_j$ as its unique deficient vertex; if $a\in K_j$, the merged component contains $a$ and has no deficient vertex.
	\end{itemize}
\end{restatable}

The next fact is what makes the round-robin step possible. If we search only from active components and mark both endpoints of a successful merge as inactive, then all augmenting paths found in one round can be applied simultaneously. This is again a static fact from Gabow.

\begin{restatable}[Parallel Compatibility of Component-Local Augmenting Paths~\cite{GABOW1995259}]{lemma}{AugmentingNoConflict}
	\label{lem:no_conflict_augment_path}
	Assume that all the augmenting paths $P_{1}, P_{2},\cdots, P_t$ produced in \cref{line:search_again} during one round are component-local augmenting paths and are selected following the active/inactive marking strategy in \cref{alg:round_robin}.
	Then, after augmenting along all these paths on \cref{line:augment_all}, $T_1, T_2, \dots, T_k$ remains a valid $k$-intersection, and $T_1, \dots, T_{k-1}$ remain spanning trees.
\end{restatable}

For completeness, we provide the proofs of \cref{lem:augment_path_structure} and \cref{lem:no_conflict_augment_path} in \Cref{sec:appendixA}.

The above lemmas show that augmenting paths can improve the solution, while the following lemma implies that if no augmenting path exists, then the current solution is already optimal.

\begin{lemma}[Reachability to Cut Certificate]
	\label{lem:reachability_cut_certificate}
	Let $s$ be the deficient vertex of a non-root $T_k$-component $K$.
	Let $L\subseteq E$ be the set of edge-vertices reachable from $s$ in the auxiliary graph $D(T)$.
	If no edge in $L$ is $T_k$-joining, then, with $S \coloneqq \{s\}\cup V(L),$ the set $\rho_G(S)$ is a rooted cut of size at most $k-1$.
\end{lemma}

\begin{proof}
	First, we show $S\subseteq V(K)$. Indeed, the set $L$ together with the vertex $s$ induces a connected subgraph on $S$: deficit-source arcs add edges entering $s$, same-head exchange arcs add edges sharing the same head as an already reached tree edge, and tree exchange arcs add fundamental paths between endpoints already reached. If $S$ contained a vertex outside $K$, then some edge of this connected subgraph would cross from $K$ to another $T_k$-component, hence would be $T_k$-joining, contradicting the assumption.

	For each $i\in[k]$, let $L_i\coloneqq L\cap T_i$. We claim that $L_i$ is a tree spanning $S$. Suppose $L_i$ were disconnected on $S$. Since $L$ connects $S$, there is an edge $f\in L\setminus T_i$ whose endpoints lie in two different components of $L_i$. The fundamental path $P_{T_i}(f)$ is well-defined: for $i<k$, $T_i$ is a spanning tree, while for $i=k$, both endpoints of $f$ lie in the same $T_k$-component $K$. By the tree exchange arcs of $D(T)$, every edge of $P_{T_i}(f)$ is reachable from $s$, and hence belongs to $L_i$. This connects the two components of $L_i$, a contradiction. Thus $L_i$ is connected on $S$, and since $L_i\subseteq T_i$, it is a tree spanning $S$.

	Therefore $T_i[S]=L_i$ for every $i\in[k]$, and so
	\[
		|T[S]|=k(|S|-1).
	\]
	Since $s$ is the unique deficient vertex in $K$ and $S\subseteq V(K)$,
	\[
		\deg_T^-(s)=k-1,
		\qquad
		\deg_T^-(v)=k \quad \forall v\in S\setminus\{s\}.
	\]
	Moreover, every vertex $v\in S\setminus\{s\}$ is the head of at least one reachable tree edge. Otherwise the number of reachable tree edges would be at most
	\[
		(k-1)+k(|S|-2)<k(|S|-1),
	\]
	contradicting $|T[S]|=k(|S|-1)$.

	Now consider any non-tree edge $g\in \rho_{E\setminus T}(S)$, and write $g=(u,v)$ with $v\in S$ and $u\notin S$. If $v=s$, then the deficit-source arc $(s,g)$ makes $g$ reachable from $s$. If $v\neq s$, choose a reachable tree edge $f\in \rho_G(v)\cap T$; then the same-head exchange arc $(f,g)$ makes $g$ reachable from $s$. In both cases $g\in L$, contradicting $u\notin S$. Hence
	\[
		\rho_{E\setminus T}(S)=\emptyset.
	\]
	It follows that
	\[
		\begin{aligned}
			|\rho_G(S)|
			 & =|\rho_T(S)|                        \\
			 & = \sum_{v\in S}\deg_T^-(v)-|T[S]|   \\
			 & = \big(k(|S|-1)+(k-1)\big)-k(|S|-1) \\
			 & = k-1.
		\end{aligned}
	\]
	Since $S\subseteq V(K)$ and $K$ does not contain $a$, we have $a\notin S$. Thus $\rho_G(S)$ is a rooted cut of size at most $k-1$.
\end{proof}

Therefore, it suffices for $\srch / \rsrch$ to maintain reachability in the auxiliary graph $D(T)$. We now turn to the only new issue in the incremental setting.

\subsection{What the Search Procedures Must Provide}
\label{sec:search_interface}

In the static setting, if the search from a deficient vertex fails, then the current graph already has a rooted cut of size less than $k$, so the algorithm may terminate permanently. In the incremental setting, this conclusion is valid only for the \emph{current} graph. A future edge insertion may create a new augmenting path. Therefore, a failed search should not be discarded; it should be viewed as a paused search state that will be resumed later.

Accordingly, the role of the procedures $\srch$ and $\rsrch$ is the following. The procedure $\srch(s)$ starts a search from a deficient vertex $s$. If it succeeds, it returns a component-local augmenting path. If it fails, it returns $\bot$ together with a set of explored edge-vertices certifying a rooted cut of the current graph of size less than $k$. The procedure $\rsrch(s,e)$ takes such a paused search and a newly inserted edge $e$, and resumes the search from the newly reachable states created by that insertion.

For the round-robin framework, we only need the following abstract guarantee.

\begin{lemma}[Correctness of the $\srch$ and $\rsrch$ procedures]
	\label{lem:search_resume_correctness}
	There exist efficient algorithms $\srch$ and $\rsrch$ such that:
	\begin{itemize}
		\item If $\srch(s)$ or $\rsrch(s,e)$ returns $\bot$, then no $T_k$-joining edge is reachable from $s$. Moreover, the algorithm has already explored all edge-vertices reachable from $s$ in the auxiliary graph $D(T)$.

		\item Otherwise, $\srch(s)$ or $\rsrch(s,e)$ returns an augmenting path $P$ starting from the deficient vertex $s$. Moreover, $\srch(s)$ always returns a component-local augmenting path (\cref{def:good_aug_path}).
	\end{itemize}
\end{lemma}

We will realize this guarantee twice. First, in \Cref{sec:brute_augment}, we give a conceptually simple brute-force implementation. Then, in \Cref{sec:cyclic_scanning}, we give the linear-time cyclic-scanning implementation used in the final algorithm.

\subsection{Warm-up: Brute-Force Pause/Resume Search}
\label{sec:brute_augment}

We first give a brute-force implementation of $\srch$ and $\rsrch$, which finally yields an incremental directed mincut algorithm with $O(mnk)$ total update time. This is not the final implementation, but it makes the pause/resume viewpoint completely explicit.

The algorithm is simply a BFS in the auxiliary graph. It takes $O(mnk)$ time because the auxiliary graph has $O(m)$ vertices, and each edge-vertex may have out-degree as large as $O(nk)$, since it may have tree exchange arcs to the edge-vertices corresponding to all edges in $P_{T_i}(e)$ for every $i \in [k]$.

In \Cref{sec:appendixB}, we compress the auxiliary graph to $O(mk\log n)$ arcs, yielding a algorithm with $O(mk^2\log^2 n)$ total update time that is still very simple.
Although this is still sub-optimal compared with the result in \Cref{sec:cyclic_scanning}, the idea is straightforward and may be easier to extend to weighted auxiliary graphs, which we hope may have further applications.

\begin{algorithm}[H]
	\caption{Brute-Force $\srch(s)$ Implementation}\label{alg:brute_reach}
	\KwGlobalVar{Graph $G=(V,E)$, a realization $T=(T_1,\cdots,T_k)$, the auxiliary graph $D(T)$;}
	\KwData{A deficient vertex $s$ in $T$\;}
	\KwResult{A component-local augmenting path $P$ from $s$ to a $T_k$-joining edge\;}

	Perform $\textsc{BFS}$ on $D(T)$ starting from $s$\;
	\Return a shortest augmenting path from $s$ to $t$, or $\bot$ if $t$ is not reachable\;
\end{algorithm}

To recover the augmenting path, we store a parent pointer at each explored vertex.
Without loss of generality, we may assume that the returned path is a component-local augmenting path. Otherwise, it suffices to replace its final edge $f$ by the deepest $T_k$-joining edge in $P_{T_i}(e)$, where $e$ is the predecessor of $f$. The resulting path remains valid and satisfies the Deepest Final Joining Edge condition in \cref{def:good_aug_path}.

We then explain how to resume this process after an incremental change.

\begin{algorithm}[H]
	\caption{Brute-Force $\rsrch(s,e)$ Implementation}
	\label{alg:brute_resume_reach}
	\KwGlobalVar{Graph $G=(V,E)$, a realization $T=(T_1,\cdots,T_k)$;}
	\KwData{An inserted edge $e\coloneqq (u,v)$\;}
	\KwResult{An augmenting path $P$ from $s$ to a $T_k$-joining edge\;}
	Update the auxiliary graph:
	$$
		\begin{aligned}
			V(D(T)) & \gets V(D(T))\cup \{e\}                                                                            \\
			A(D(T)) & \gets A(D(T))\cup \mathbf 1[e\text{ is not }T_k\text{-joining}]\{(e,f) : f \in P_{T_i}(e), i \in [k]\} \\
            &\quad\cup \{(f,e) : f \in \rho_G(v) \cap  T\} \\
            &\quad \cup \mathbf 1[\mathrm{head}(e)\text{ is deficient}]\{(\mathrm{head}(e),e)\}\\
&\quad\cup \mathbf 1[e\text{ is }T_k\text{-joining}]\{(e,t)\};
		\end{aligned}
	$$

	\If{$e$ has a predecessor in $D(T)$ explored in previous $\srch/\rsrch$}{
		Resume a reachability search on $D(T)$ starting from $e$\;
		\Return an augmenting path from $s$ to a $T_k$-joining edge-vertex $f$, or $\bot$ if no such $f$ is reachable\;
	}

	\Return $\bot$\;
\end{algorithm}

When a new edge \(e\) arrives, we first check whether any predecessor of its edge-vertex in the auxiliary graph has already been explored. If so, the edge-vertex \(e\) is reachable from \(s\) as well. We then resume the reachability search from \(e\), without revisiting vertices explored in previous \(\srch\)/\(\rsrch\) calls.

\begin{theorem}[Brute-force implementation of $\srch/\rsrch$]
	\label{thm:bruteforce_search_resume}
	The procedures $\srch$ and $\rsrch$ can be implemented by explicit BFS on the auxiliary graph so that they satisfy \cref{lem:search_resume_correctness}.

	Moreover, for a fixed value of $k$, the total time spent on all calls to $\srch$ and $\rsrch$ during one iteration of the \texttt{while}-loop in \cref{line:while_non_full} of \cref{alg:round_robin} (i.e. one round) is $O(mnk)$.
\end{theorem}

\begin{proof}
	\textbf{Reachability.}
		$\srch(s)$ performs BFS in the current auxiliary graph $D(T)$ starting from $s$. Therefore, if it returns $\bot$, the set of explored edge-vertices is exactly the set of edge-vertices reachable from $s$, and no $T_k$-joining edge is reachable.

		Now consider $\rsrch(s,e)$ after a new edge $e=(u,v)$ is inserted. The only new auxiliary vertex is the edge-vertex $e$, and every new arc is incident to it. Thus, if none of its predecessors in $D(T)$ has already been explored, the reachable region from $s$ does not grow. Otherwise, $\rsrch$ resumes the reachability search from $e$ and never revisits previously explored vertices. Hence it explores exactly the newly reachable part of $D(T)$.

	Consequently, whenever $\rsrch(s,e)$ returns $\bot$, the maintained explored set is the set of edge-vertices reachable from $s$ in $D(T)$, and no $T_k$-joining edge is reachable.

	\textbf{Component-local augmenting path.} We now show that whenever $\srch$ returns an augmenting path, it is component-local. Because $\srch$ uses BFS, it finds a shortest augmenting path (and modifying the final edge does not change its length). This guarantees the Treewise Shortcut-Free and No Premature Joining Edge properties.
	The final modification step guarantees the Deepest Final Joining Edge property.

	\textbf{Runtime.} The time is bounded by the size of the auxiliary graph $D(T)$, which is $O(mnk)$.
\end{proof}

\subsection{Incremental Cyclic Scanning}
\label{sec:cyclic_scanning}

We now return to Gabow's actual search strategy and show that the same pause/resume idea can be implemented in linear total time per round within a fixed stage.

We first recall the cyclic scanning algorithm for static $k$-intersection from \cite{GABOW1995259} and then modify it for the incremental setting. We denote $L_i$ as the set containing the deficient vertex $s$ and all endpoints of explored edges in $T_i$ during the search process. We also write $T_{((i-1)\bmod k) + 1}, L_{((i-1)\bmod k)+1}$ as $T_i,L_i$ for short in this section.

The key insight of the cyclic scanning algorithm is that, for an edge $e \in T_{i-1}$, we do not need to scan all successors $f \in P_{T_j}(e)$ for every $j \in [k]$. Instead, it suffices to scan only the edges  $P_{T_{i}}(e)$ in the next layer. The reason is that $L_i$ then inherits the connectivity information carried by $e$, and propagates it further to $L_{i+1}, L_{i+2}, \cdots$. Indeed, $P_{T_{i+1}}(e)\subseteq \bigcup_{f \in P_{T_{i}}(e)} P_{T_{i+1}}(f)$. Intuitively, $P_{T_{i}}(e)$ is the $T_i$-path between both endpoints of $e$ (note that $P_{T_{i}}(e)$ exists because otherwise we find a joining edge), so $\bigcup_{f \in P_{T_{i}}(e)} P_{T_{i+1}}(f)$ is a connected subtree of $T_{i+1}$ that contains both endpoints of $e$, so it must contain $P_{T_{i+1}}(e)$, which is the minimal subtree satisfying this. In other words, every edge in $P_{T_{i+1}}(e)$ that is directly reachable from $e$ is now reached indirectly through some edge in $P_{T_{i}}(e)$.

The entire search process works in \emph{phases}. We maintain a queue $Q$ of explored edges and an index $i$, which indicates that the current phase is scanning exchanges into $T_i$. Initially, we set $i\gets 1$, explore all edges in $\rho_G(s)\setminus T$, and put them into $Q$. Whenever the edge popped from $Q$ lies in $T_i$, we advance to the next phase by increasing $i\gets i+1$.

At the beginning of the phase for $T_i$, we maintain the invariant that the explored edges in $T_i$ form a connected subtree $L_i$ containing the root of the search. Thus, for a popped edge $e=(u,v)$, if both endpoints already lie in $L_i$, then scanning $P_{T_i}(e)$ would reveal nothing new, and we may skip it. Otherwise, the invariant and scanning order guarantee that exactly one endpoint of $e$ lies in $L_i$, and the fundamental path $P_{T_i}(e)$ contains a unique suffix of unexplored edges leaving $L_i$. We scan this suffix in the order in which it leaves $L_i$, explore these newly discovered edges, and enqueue them. Whenever such a newly explored edge $f$ has head $w$, we also explore all unexplored edges in $\rho_G(w)\setminus T$ and enqueue them; these steps follow same-head exchange arcs in the auxiliary graph.

If during this process we encounter an edge joining $T_k$, then by tracing back the labels, we obtain an augmenting path. Otherwise, when the queue becomes empty, no augmenting path exists.

Each newly explored edge in $T_{i-1}$ is used only to explore edges in $T_{i}$, so reachability information propagates cyclically through $T_ 1, T_2,\dots, T_k$ rather than branching to all $k$ forests at once. This avoids an extra factor of $k$ in the running time. Moreover, thanks to the invariant and scanning order, one can efficiently find all unexplored edges on $P_{T_i}(e)$ without using advanced data structures; a discussion of this `find unexplored edges on $P_{T_i}(e)$' step is given in \cref{sec:scan_C}.

In the incremental setting, the only new question is whether these same invariants survive when the search is paused, a new edge is inserted, and the search later resumes. We will show that they do.

Formally, the pseudocode is shown below.

\begin{algorithm}[H]
	\caption{$\srch(s)$ Cyclic Scanning for $k$-intersection}\label{alg:reach}
	\KwGlobalVar{Graph $G=(V,E)$, a realization $T=(T_1,\cdots,T_k)$;}
	\KwData{A deficient vertex $s$ in $T$\;}
	\KwResult{A component-local augmenting path $P$ from $s$ to a $T_k$-joining edge\;}

	Let $Q\gets$ an empty queue, $i\gets 1$\;
	\For{$e \in \rho_G(s)\setminus T$}{
		$\expl(e,s)$ and return the path if it returns a path\;
	}
	\tcp{Main Loop,
		maintain $L_i\coloneqq \{s\} \cup  V( e: e \text{ is explored}, e\in  T_i)$
	}
	\While{$Q$ is not empty}{
		$e=(u,v) \gets Q.\mathrm{pop}()$\;
		\lIf{$e \in T_i$}{
			$i\gets i+1$\label{line:inc_i}
		}
		\lIf{both $u,v \in L_i$}{
			\Continue
		}
		\For{each unexplored $f \in P_{T_i}(e)$ (in the order of leaving $L_i$)}{\label{line:scan_C} \tcp{implemented by the data structure in \cref{thm:scan_C}}
			$\expl(f,e)$ and return the path if it returns a path\;

			Let $w \gets $ the head of $f$\;
			\If{edges in $\rho_G(w)\setminus T$ are not explored}{
				\For{$h : h \in \rho_G(w)\setminus T$}{
					$\expl(h,f)$ and return the path if it returns a path\;

				}
			}
		}
	}
	\Return $\bot$\;
\end{algorithm}

In the main loop, scanning $P_{T_i}(e)$ follows tree exchange arcs. After exploring a new edge $f$ this way, we scan the unexplored edges in $\rho_G(\mathrm{head}(f))\setminus T$; this follows same-head exchange arcs.

The $\expl$ procedure is defined as follows. Intuitively, $\expl(e,\pi)$ explores the edge-vertex $e$ and labels $\pi$ as its parent. If $e$ is $T_k$-joining, it also extracts an augmenting path from the search tree.

\begin{algorithm}[H]
	\caption{$\expl(e,\pi)$}\label{alg:label_and_push}
	\KwGlobalVar{The current search labels and queue $Q$;}
	\KwData{An edge $e$ and its predecessor $\pi$ in the search tree;}
	\KwResult{An augmenting path ending at $e$, or $\bot$;}
	Label $e$ with $\pi$\;
	\If{$e$ is $T_k$-joining}{
		\Return the path from $s$ to $e$ recovered using the labels (possibly after modifying the final joining edge to satisfy the Deepest Final Joining Edge condition)\;
	}
	$Q.\mathrm{push}(e)$\;
	\Return $\bot$\;
\end{algorithm}

We then explain how to resume this process after an incremental change.

\begin{algorithm}[H]
	\caption{$\rsrch(s,e)$ resume Cyclic Scanning for $k$-intersection}
	\label{alg:resume_reach}
	\KwGlobalVar{Graph $G=(V,E)$, a realization $T=(T_1,\cdots,T_k)$;}
	\KwData{An inserted edge $e\coloneqq (u,v)$\;}
	\KwResult{An augmenting path $P$ from $s$ to a $T_k$-joining edge\;}

	Let $\pi$ be $s$ if $v=s$, and otherwise any explored edge in $\rho_G(v)\cap T$, if one exists\;
	\lIf{$\pi=\bot$}{\Return $\bot$}
	$\expl(e,\pi)$ and return the path if it returns a path\;
	Run the Main Loop as in \cref{alg:reach}\;
\end{algorithm}

We begin with the static invariants from Gabow. 

\begin{lemma}[Invariants in static Cyclic Scanning~\cite{GABOW1995259}]
	\label{lem:cyc_scan_invariant}
	At the beginning of each phase (i.e., immediately after executing \cref{line:inc_i}), the following hold:
	\begin{enumerate}
		\item $L_i$ is a connected subtree of $T_i$;
		\item $L_i \subseteq L_{i+1} \subseteq \cdots \subseteq L_{i+k-1}\subseteq L_i \cup V(Q)$;
		\item $Q \subseteq T_{i-1} \cup (E\setminus T)$, and every edge of $Q$ has an endpoint either in $L_i$ or in some earlier edge of $Q$.
	\end{enumerate}
\end{lemma}

\begin{proof}
	We prove the invariants by induction on the number of phases. 

	\textbf{Base Case:} $i=1$, and $Q$ is populated with edges in $\rho_G(s) \setminus T$. For all $j \in [k]$, $L_j = \{s\}$.
	(1) Each $L_j$ is a single vertex, which is trivially a connected subtree.
	(2) Since all $L_j$ are identical, the inclusions $L_1 \subseteq \dots \subseteq L_k$ hold, and $L_k \subseteq L_1 \cup V(Q)$ is immediate.
	(3) Every edge $e \in Q$ enters $s$ (so it has an endpoint in $L_1$), and $Q$ only contains edges in $E \setminus T$.

	\textbf{Inductive Step:} Assume the invariants hold at the start of phase $i$. We show they remain valid at the start of phase $i+1$ (after $i$ is incremented at \cref{line:inc_i}).

	\begin{itemize}
		\item \textbf{Invariant (1):} We show this holds throughout phase $i$. When the algorithm pops an edge $e$, it explores two cases:
		      \begin{enumerate}
			      \item All edges $f$ on the path of unexplored edges in $P_{T_i}(e)$: By Invariant (3), $e$ has at least one endpoint in the current $L_i$. The algorithm grows $L_i$ by adding a path with one end already attached to $L_i$; thus, $L_i$ remains a connected subtree. The \texttt{if} condition at \cref{line:scan_C} ensures we only process $e$ if it actually expands $L_i$.
			      \item Edges $h \in E \setminus T_i$: These edges do not belong to $T_i$ and thus do not modify $L_i$.
		      \end{enumerate}
		      Applying this argument to each popped edge $e$ proves the statement.

		\item \textbf{Invariant (3):} Let $Q'$ denote the set of elements enqueued during phase $i$. $Q \subseteq T_{i-1} \cup (E\setminus T)$ holds trivially when $k=1$, so we assume $k>1$. By the inductive hypothesis, no edges in $Q$ belong to $T_i$, and the first edge enqueued in $Q'$ is from $T_i$. Thus, the algorithm processes exactly all edges in $Q$ before the phase ends. Consequently, $Q$ becomes $Q'$ at the start of phase $i+1$.

		      Moreover, $Q'$ consists only of previously unexplored edges from $P_{T_i}(e)$ (which belong to $T_i$) or edges from $\rho_G(w) \setminus T$ (which belong to $E \setminus T$). Thus, $Q' \subseteq T_i \cup (E \setminus T)$ holds for the next phase.

		      Finally, any edge $f \in Q'$ satisfies the endpoint property:
		      \begin{enumerate}
			      \item If $f \in P_{T_i}(e)$, it is enqueued because $e$ has an endpoint in $L_i$. Since $L_i \subseteq L_{i+1}$ (by Invariant (2) for phase $i$), this endpoint is also in $L_{i+1}$. Edges are enqueued in their order along the path, ensuring each overlaps with the previous one or the endpoint of $e$.
			      \item If $h \in E \setminus T_i$ is enqueued, it shares an endpoint $w$ with some $f \in P_{T_i}(e)$ that was just explored.
		      \end{enumerate}

		\item \textbf{Invariant (2):} Let $L'_i$ and $Q'$ be the states of $L_i$ and $Q$ at the start of phase $i+1$. Note that $L_j$ for $j \neq i$ remain unchanged during this phase, so their relative inclusions persist.

		      For each $e \in Q$, we explore and enqueue all previously unexplored edges on $P_{T_i}(e)$. This implies $L_i \cup V(Q) \subseteq L'_i$. Furthermore, we have $L'_i \subseteq L_{i+1} \cup V(Q')$ because $L_i \subseteq L_{i+1}$ at the start of phase $i$, and every new vertex added to $L_i$ to form $L'_i$ corresponds to an edge enqueued in $Q'$.

		      Combining the inclusion chain from phase $i$ with these new relations:
		      \[
			      L_{i+1} \subseteq L_{i+2} \subseteq \cdots \subseteq L_{i+k-1} \subseteq L_i \cup V(Q) \subseteq L'_i \subseteq L_{i+1} \cup V(Q')
		      \]
		      Removing the intermediate terms confirms the invariant for phase $i+1$.
	\end{itemize}
	This completes the induction.
\end{proof}

The next lemma is the new incremental step. It shows that these same invariants still hold after a paused search is resumed following an insertion.

\begin{lemma}
	\label{lem:resume_invariant}
	The invariants in \cref{lem:cyc_scan_invariant} are maintained after \cref{alg:resume_reach} enqueues $e$ and invokes the Main Loop.
\end{lemma}

\begin{proof}
	We first verify that the invariants hold immediately after exploring and enqueueing the newly inserted edge $ e = (u,v)$, and before resuming the Main Loop. Assume that before inserting $e=(u,v)$, the paused state satisfies all the invariants of \cref{lem:cyc_scan_invariant}. Then, by invariant~(2) and the fact that $Q$ is currently empty, we have
	\begin{equation}
		\label{eq:invariant_before_insert}
		L_i = L_{i+1} = \cdots = L_{i+k-1}.
	\end{equation}

	\begin{itemize}
		\item \textbf{Invariant~(3).} If $v=s$, then $v\in L_i$ by definition. Otherwise, the algorithm resumes the Main Loop only when there exists an explored tree edge $f\in \rho_G(v)\cap T$, and hence $v\in L_i$ by \eqref{eq:invariant_before_insert}. Since $Q$ now contains only the edge $e=(u,v)$, invariant~(3) holds.
		\item \textbf{Invariant~(1).} Since $e \notin T$, exploring and enqueueing $e$ does not alter any $L_i$. Thus, all $L_i$ remain subtrees of their respective $T_i$.
		\item \textbf{Invariant~(2).} Again, since $e \notin T$, the chain of inclusions
		      \[
			      L_i \subseteq L_{i+1} \subseteq \cdots \subseteq L_{i+k-1}
		      \]
		      remains unchanged. The final inclusion $L_{i+k-1} \subseteq L_i \cup V(Q)$ also holds, since $L_{i+k-1}=L_i$ by \eqref{eq:invariant_before_insert} and $Q$ contains only $e$.
	\end{itemize}

	After these updates, the resumed execution processes $e$ in the same way as an edge that is newly explored and enqueued during a normal execution of the search. Since the state at the moment the Main Loop resumes satisfies all three invariants, the inductive proof of \cref{lem:cyc_scan_invariant} applies directly to the remainder of the execution.

	Thus, the invariants hold throughout the resumed search.
\end{proof}

With these invariants in place, the correctness and running time of incremental cyclic scanning now follow.

\begin{theorem}[Cyclic-scanning implementation of $\srch/\rsrch$]
	\label{thm:cyclic_search_resume}
	The procedures $\srch$ and $\rsrch$ can be implemented via cyclic scanning so that they satisfy \cref{lem:search_resume_correctness}.

	Moreover, for a fixed value of $k$, the total time spent on all calls to $\srch$ and $\rsrch$ during one round, i.e., one iteration of the \texttt{while}-loop in \cref{line:while_non_full} of \cref{alg:round_robin}, is $O(m)$.
\end{theorem}

\begin{proof}[Proof of \cref{thm:cyclic_search_resume}]
	\textbf{Reachability.}
	We show that cyclic scanning explores all edge-vertices reachable from $s$ in the auxiliary graph $D(T)$.

	Every call to $\expl$ is made through a valid predecessor in $D(T)$: deficit-source arcs are used at initialization, same-head exchange arcs are used when scanning $\rho_G(w)\setminus T$ after an explored tree edge with head $w$, and tree exchange arcs are used when scanning $P_{T_i}(e)$. Thus all explored edge-vertices are reachable.

	Conversely, suppose $\srch(s)$ stops with $Q=\emptyset$. By \cref{lem:cyc_scan_invariant}, at termination we have $L_1=L_2=\cdots=L_k\eqqcolon S.$ Also, every explored edge has both endpoints in $S$: once an explored edge is popped, the scan either skips it because both endpoints are already in the current $L_i$, or scans the corresponding path in $T_i$ and adds the other endpoint.

	Assume for contradiction that some edge reachable from $s$ is unexplored, and let $g$ be the first unexplored edge on a shortest reachable path. Let $\pi$ be its predecessor.
    
	If $\pi=s$, then $g\in \rho_G(s)\setminus T$, so $g$ is explored during initialization. If $\pi$ is a tree edge and $(\pi,g)$ is a same-head exchange arc, then $\pi$ and $g$ have the same head $w$; when $\pi$ was explored, the algorithm scanned $\rho_G(w)\setminus T$ and therefore explored $g$.

	If $(\pi,g)$ is a tree exchange arc, then for some $j\in[k]$ we have $\pi\notin T_j$, $g\in T_j$, and $g\in P_{T_j}(\pi)$. Since both endpoints of $\pi$ lie in $S=L_j$ and $L_j$ is a connected subtree of $T_j$, the entire tree path in $T_j$ between the endpoints of $\pi$ is contained in $L_j$. Hence every edge of $P_{T_j}(\pi)$, including $g$, is explored, a contradiction.

	Therefore, when $\srch(s)$ returns $\bot$, the set of explored edge-vertices is exactly the set of edge-vertices reachable from $s$ in the current auxiliary graph, and no $T_k$-joining edge is reachable.

	Now consider $\rsrch(s,e)$ after inserting a new edge $e=(u,v)$. Assume inductively that before the insertion, the set of explored edge-vertices is exactly the set of edge-vertices reachable from $s$. Since the realization $T$ has not changed while the search is paused, inserting $e$ only adds the new edge-vertex $e$ and arcs incident to $e$ in the auxiliary graph; hence, the reachable region can grow only if $e$ itself becomes reachable.

	Since $e\notin T$, the possible already-explored predecessors of $e$ are $s$ (when $v=s$), and explored tree edges in $\rho_G(v)\cap T$, via same-head exchange arcs. These cases are checked in \cref{alg:resume_reach}. If neither case holds, then $e$ is not reachable, and the old explored set remains the full reachable set. If one case holds, $\expl$ explores $e$ with a valid predecessor. By \cref{lem:resume_invariant}, after enqueueing $e$ the same cyclic-scanning invariants hold before the Main Loop is invoked. The argument above then applies directly to the resumed Main Loop: when it stops with $Q=\emptyset$, all edge-vertices reachable through $e$ have been explored.

	Thus, whenever $\rsrch(s,e)$ returns $\bot$, the set of explored edge-vertices is exactly the set of edge-vertices reachable from $s$ in the current auxiliary graph, and no $T_k$-joining edge is reachable.

	\textbf{Component-local augmenting path.} We now show that whenever $\srch$ returns an augmenting path, it is component-local. When the algorithm processes an edge $e$, every unexplored edge of $P_{T_i}(e)$ is assigned $e$ as its parent. Hence at most one of these edges can occur after $e$ on the returned search-tree path, which gives the Treewise Shortcut-Free property.
	Since the algorithm halts as soon as it encounters a $T_k$-joining edge, the returned path satisfies the No Premature Joining Edge property. The final modification step guarantees the Deepest Final Joining Edge property.

	\textbf{Runtime.}
	As searches start from disjoint $T_k$-components and stop as soon as they reach a $T_k$-joining edge, each edge that is not $T_k$-joining is reached only by the search from its own component, while each $T_k$-joining edge is reached only once before its endpoint components are deactivated.
	Since each edge is explored and scanned by the algorithm at most twice (once during the paused/resumed search, and once after the exploration states are reset in \cref{line:search_again}), and the data structure for scanning $P_{T_i}(e)$ runs in time linear in the number of edges scanned, the total runtime for reachability search per round is $O(m)$.

    Finally, we check the initialization of the data structures of \cref{thm:scan_C}. At the beginning of each round, we initialize the data structures once for each tree $T_1,\ldots,T_{k-1}$ (rooted at $a$), in $O((k-1)n)$ time, and once across all $T_k$-components (rooted at their deficient vertices), in $O(n)$ time because the components are vertex-disjoint. These structures are reused throughout the round. We reset a data structure at \cref{line:search_again} and when switching components; each $T_k$-component is reset at most twice, at cost linear in its explored set. Thus initialization and resetting cost $O(kn)=O(m)$ per round, since $m\ge n$ and the existing $(k-1)$-intersection gives $m\ge (k-1)(n-1)$.
\end{proof}

\subsection{Proof of the Main Theorem}
\label{sec:proof_main_incremental}

We can now return to the round-robin framework and finish the proof of \Cref{thm:main}.

\begin{proof}[Proof of \cref{thm:main}]
	\textbf{Runtime.} We analyze the number of rounds, namely the iterations of the \texttt{while} loop on \cref{line:while_non_full}. Assume that the previous stages have already produced a valid $(k-1)$-intersection.

	In each round, every active $T_k$-component initiates a search. Since each successful search marks at most two $T_k$-components as inactive, the \texttt{if} condition on \cref{line:check_active} is satisfied for at least $\lceil (q-1)/2 \rceil$ components, producing at least $\lceil (q-1)/2 \rceil$ component-local augmenting paths. By \cref{lem:no_conflict_augment_path}, augmenting along all these paths simultaneously maintains a valid realization and successfully merges pairs of $T_k$-components.

	Thus, the number of $T_k$-components drops by at least a factor of two per round, which means the \texttt{while} loop on \cref{line:while_non_full} is executed at most $O(\log n)$ times for a fixed $k$. The total runtime is bounded by:

	$$
		T \le
		\underbrace{k}_{\text{the } k \text{-loop}}
		\times
		\underbrace{O(\log n)}_{\text{the \texttt{while} loop}}
		\times
		\underbrace{O(m)}_{\text{all $\srch,\rsrch$ calls for a fixed $k$ (\cref{thm:cyclic_search_resume})}}
		= O(km \log n).
	$$

	\textbf{Correctness.}
	Whenever the algorithm reports the answer at \cref{line:report_k_minus_1}, it is because either $\srch(s)=\bot$ or, after some edge insertions, $\rsrch(s,e)=\bot$. Let $L$ be the set of explored edge-vertices maintained by this paused search. By \cref{lem:search_resume_correctness}, $L$ is exactly the set of edge-vertices reachable from $s$ in the current auxiliary graph $D(T)$, and no $T_k$-joining edge is reachable from $s$. Applying \cref{lem:reachability_cut_certificate} gives a rooted cut of size at most $k-1$. By \cref{lem:edmonds}, the current graph does not contain any complete $k$-intersection.

	On the other hand, since the algorithm has already constructed a complete $(k-1)$-intersection, the minimum rooted connectivity is exactly $k-1$. Running the algorithm symmetrically on both $G$ and $G^\mathrm{rev}$ yields the correct directed global minimum cut.

	Regarding explicit maintenance of the minimum cut, the search procedure maintains the graph-vertex set $\{s\}\cup V(L)$. This set is updated only by inserting a vertex or by resetting it to $\emptyset$ when a search finds an augmenting path and restarts from another active $T_k$-component. Between resets, each vertex is added at most once. Whenever a vertex $v$ is added, we update the maintained edges of $\rho_G(\{s\}\cup V(L))$ by adding $\rho_G(v)$ and removing internal edges incident to $v$. This takes $O(m)$ total time per round and preserves the overall asymptotic runtime.
\end{proof}

\bibliography{refs}

\appendix
\section{Omitted Proofs}
\label{sec:appendixA}

\AugmentingPathStructure*
\begin{proof}
	Let $P=(s,e_1,\ldots,e_r,t)$ be a component-local augmenting path starting in the $T_k$-component $K_i$.

	\paragraph{Locality before the final edge.}
	We prove inductively that every $e_j$ with $j<r$ has both endpoints in $V(K_i)$. If $r>1$, the deficit-source arc $(s,e_1)$ gives $\mathrm{head}(e_1)=s\in V(K_i)$, and $e_1$ is not $T_k$-joining by the No Premature Joining Edge property; hence its tail also lies in $V(K_i)$. Now suppose that $e_j$ lies inside $K_i$ and $j+1<r$. If $(e_j,e_{j+1})$ is a same-head exchange arc, then $e_{j+1}$ has its head in $V(K_i)$ and is not $T_k$-joining, so its tail also lies in $V(K_i)$. If it is a tree exchange arc, then $e_{j+1}\in P_{T_h}(e_j)$ for some $h\in[k]$. The No Premature Joining Edge property says that this fundamental path contains no $T_k$-joining edge. Since the endpoints of $e_j$ lie in $V(K_i)$, the entire path, and in particular $e_{j+1}$, lies inside $K_i$.

	\paragraph{The final joining edge.}
	The joining-sink arc $(e_r,t)$ shows that $e_r$ is $T_k$-joining. It remains to show that $e_r$ is incident to $V(K_i)$. If $r=1$, this follows from $\mathrm{head}(e_r)=s$. If the arc entering $e_r$ is a same-head exchange arc, then $\mathrm{head}(e_r)=\mathrm{head}(e_{r-1})\in V(K_i)$. Otherwise it is a tree exchange arc, so $e_r\in P_{T_h}(e_{r-1})$ for some $h<k$. Along each of the two paths from an endpoint of $e_{r-1}$ to their LCA in $T_h$, edge depths strictly decrease. Consequently, a maximum-depth $T_k$-joining edge on $P_{T_h}(e_{r-1})$ is one of the first joining edges encountered from an endpoint of $e_{r-1}$ and is therefore incident to $V(K_i)$. The Deepest Final Joining Edge condition chooses such an edge. Thus, the other endpoint of $e_r$ lies in a different $T_k$-component $K_j$.

	\paragraph{Validity and component structure.}
	For each $h\in[k]$, list the exchanges in $T_h$ in their order along $P$. The Treewise Shortcut-Free property says that no edge deleted by a later exchange lies on the fundamental path used by an earlier exchange. Hence \cref{lem:matroid_exchange} shows that every $T_h$ remains a forest. Each $T_h$ with $h<k$ also retains $n-1$ edges and therefore remains a spanning tree. Every exchange that modifies $T_k$ uses only edges inside $K_i$ and preserves their number, so the $T_k$-component partition is unchanged until $e_r$ is inserted. Since $e_r$ is $T_k$-joining, inserting it into $T_k$ merges exactly $K_i$ and $K_j$ and preserves acyclicity.

	\paragraph{Deficient vertices.}
	Under the index shift defining augmentation, $e_1$ is the sole unpaired addition. Every edge that leaves the union is paired, across a same-head exchange arc, with an edge of the same head that enters it. Hence $\deg_T^-(s)$ increases by one and every other indegree is unchanged. Thus $s$ is no longer deficient, and the merged component has the deficient vertex of $K_j$ as its unique deficient vertex. If $a\in K_j$, the merged component contains $a$ and has no deficient vertex.
\end{proof}

\AugmentingNoConflict*

\begin{proof}
	Fix the realization at the beginning of the round, and index the selected paths by $\ell\in[t]$. Write $P^\ell=(s^\ell,e_1^\ell,e_2^\ell,\ldots,e_{r^\ell}^\ell,t)$, and let $K^\ell$ be its starting $T_k$-component.

	By \cref{lem:augment_path_structure}, each path stays inside its starting component until its final joining edge. Once a path is selected, both components joined by that edge become inactive, so no later path can start in either component or use the same edge. Hence the selected paths are edge-disjoint.

	By the exchange view of augmentation in \cref{fig:augmentation-exchanges}, we can augment along these paths by first applying the exchange $T_i\gets T_i-e_{j+1}^\ell+e_j^\ell$ to the appropriate tree $T_i\ni e_{j+1}^\ell$ for every tree-exchange arc $(e_j^\ell,e_{j+1}^\ell)$, and then inserting each final $T_k$-joining edge $e_{r^\ell}^\ell$ into $T_k$.

	\paragraph{The tree exchanges.}
	We first show that all tree exchanges induced by these paths can be combined without conflict. Fix $T_i$. For each $\ell\in[t]$ and each tree-exchange arc $(e_j^\ell,e_{j+1}^\ell)$ with $e_{j+1}^\ell\in T_i$, call the corresponding exchange \emph{final} if $j+1=r^\ell$. Order these exchanges as follows: first list all non-final exchanges, preserving their order within each path and ordering different paths arbitrarily; then list the final exchanges by increasing $d(e_{r^\ell}^\ell)$, breaking ties by decreasing edge identifier (as in \cref{def:good_aug_path}).

	We verify that this order satisfies the shortcut-free condition in \cref{lem:matroid_exchange}. First consider a non-final exchange $(e_j^\ell,e_{j+1}^\ell)$. By the \emph{No Premature Joining Edge property}, $P_{T_i}(e_j^\ell)$ contains no $T_k$-joining edge and lies entirely inside $K^\ell$. By the Treewise Shortcut-Free property, no edge deleted later by the same path lies on $P_{T_i}(e_j^\ell)$. Any edge deleted by a later non-final exchange from another path lies inside a different starting component, while any edge deleted by a later final exchange is $T_k$-joining. Thus neither type of edge can belong to $P_{T_i}(e_j^\ell)$.

	Next consider two final exchanges in $T_i$: $(e_{r^\ell-1}^\ell,e_{r^\ell}^\ell)$, followed by $(e_{r^h-1}^h,e_{r^h}^h)$. If $e_{r^h}^h\in P_{T_i}(e_{r^\ell-1}^\ell)$, then, by the chosen order, $e_{r^h}^h$ is either deeper than $e_{r^\ell}^\ell$ or equally deep with a smaller identifier. Either case contradicts the \emph{Deepest Final Joining Edge property} in \cref{def:good_aug_path}. Hence $e_{r^h}^h\notin P_{T_i}(e_{r^\ell-1}^\ell)$.

	By \cref{lem:matroid_exchange}, every $T_i$ remains a forest; for $i<k$, it still has $n-1$ edges and is therefore a spanning tree. All exchanges in $T_k$ are internal to original components and preserve their edge counts, so these components remain trees before the joining edges are inserted.

	\paragraph{Adding the Joining Edges.}
	It remains to show that adding the final joining edges does not create a cycle in $T_k$. Let $H$ be the undirected graph whose vertices are the original $T_k$-components and whose edges are the pairs of components joined by $e_{r^1}^1,\ldots,e_{r^t}^t$. When $P^\ell$ is selected, its starting component $K^\ell$ is active, while both endpoint components of every earlier joining edge are inactive. Thus $K^\ell$ is isolated in the subgraph formed by the earlier edges, so adding $e_{r^\ell}^\ell$ cannot create a cycle. Hence $H$ is a forest.	Since every original $T_k$-component is a tree and contracting these components yields $H$, the new $T_k$ is a forest.

	Finally, each tree of $H$ has exactly one non-starting component, i.e., a component from which no selected augmenting path starts. By \cref{lem:augment_path_structure}, augmentation removes the deficiency of every starting component and changes no other indegree. Hence each merged component has the deficient vertex of its non-starting component, or none if that component is the root component. Thus both the indegree bounds and the component invariant are preserved.
\end{proof}

\section{Auxiliary-Graph Compression for Brute-Force Search}
\label{sec:appendixB}

In this section, we give a simple incremental directed minimum cut algorithm with $O(mk^2\log^2 n)$ total update time. This is slightly worse than the bound in \Cref{thm:main}, but remains near-linear for polylogarithmic $k$.

The algorithm is based on the brute-force algorithm from \Cref{sec:brute_augment}, but compresses the arcs of the auxiliary graph while preserving reachability and distances. The idea is to represent the tree exchange arcs from $e$ to the edges in $P_{T_i}(e)$ using virtual vertices and binary lifting. Recall that $V(D(T))=V\cup E\cup\{t\}$.

\begin{lemma}
	One can build a 0/1-weighted digraph $D'(T)$ such that $V(D(T))\subseteq V(D'(T))$, $|V(D'(T))|=O(nk\log n+m)$, and $|A(D'(T))|=O(nk\log n+mk)$. We call the vertices in $V(D'(T))\setminus V(D(T))$ \emph{virtual vertices}.

	Moreover, for every $v\in V$:

	\begin{itemize}
		\item If there is a path $P$ in $D(T)$ from $v$ to $t$ of length $l$, then there is a path $P'$ in $D'(T)$ from $v$ to $t$ of length $l$, and removing all virtual vertices from $P'$ yields $P$.

		\item If there is a path $P'$ in $D'(T)$ from $v$ to $t$ of length $l$, then removing all virtual vertices from $P'$ yields a valid path $P$ in $D(T)$ from $v$ to $t$ of length $l$.
	\end{itemize}
\end{lemma}

\begin{proof}
	For $D'(T)$, the length of a path means its total arc weight.
	Root each $T_i$; for $i=k$, root each $T_k$-component separately. Write $p_i(x)$ for the parent of $x$, $d_i(x)$ for its depth, and $\mathrm{edge}_i(x)$ for the edge joining $x$ to $p_i(x)$.
	The following construction is similar to a sparse table: virtual vertices represent power-of-two blocks of root paths.

	For every $i\in[k]$, vertex $x$, and integer $j\geq 0$ with $2^j\leq d_i(x)$, create a virtual vertex $x_{i,j}$. Add
	\[
		\begin{aligned}
		x_{i,0} &\longrightarrow \mathrm{edge}_i(x)
		&&\text{with weight $1$},\\
		x_{i,j} &\longrightarrow x_{i,j-1},
		&x_{i,j}&\longrightarrow (p_i^{2^{j-1}}(x))_{i,j-1}
		&&\text{with weight $0$ for $j\geq 1$}.
		\end{aligned}
	\]
	By induction on $j$, $x_{i,j}$ reaches exactly the $2^j$ edge-vertices on the path from $x$ to $p_i^{2^j}(x)$, each at total cost $1$.

	Now fix $e=(u,v)\notin T_i$ for which $P_{T_i}(e)$ is defined, and let $a=\mathrm{LCA}_i(u,v)$. For each $x\in\{u,v\}$, let $\ell_x=d_i(x)-d_i(a)$. If $\ell_x>0$, set $q_x=\lfloor\log \ell_x\rfloor$ and add the weight-$0$ arcs
	\[
		(e,x_{i,q_x})
		\quad\text{and}\quad
		(e,(p_i^{\ell_x-2^{q_x}}(x))_{i,q_x}),
	\]
	omitting one if they coincide. As in a sparse table, the two length-$2^{q_x}$ blocks cover the entire $x$--$a$ branch because $2^{q_x}\leq\ell_x<2^{q_x+1}$. Thus $e$ needs at most two weight-$0$ arcs for each branch, and at most four to represent all tree exchange arcs from $e$ to $P_{T_i}(e)$.
	\Cref{fig:binary-lifting-compression} illustrates the LCA split and the two overlapping blocks.

	\begin{figure}[H]
		\centering
		\begin{tikzpicture}[
				x=1cm,y=1cm,
				treevertex/.style={circle,fill=black,inner sep=1.35pt},
				treeedge/.style={line width=.8pt,gray!70},
				originaledge/.style={-{Stealth[length=1.8mm]},thick,dashed},
				edgevertex/.style={rectangle,draw,rounded corners=1pt,fill=gray!8,minimum width=6mm,minimum height=4.5mm,inner sep=1pt},
				virtual/.style={rectangle,draw=blue!65!black,rounded corners=1.5pt,fill=blue!6,minimum height=4.7mm,inner sep=1.5pt,align=center},
				tablevertex/.style={virtual,minimum width=8.5mm,minimum height=5.2mm,inner sep=1pt,font=\scriptsize},
				selected/.style={tablevertex,draw=blue!80!black,fill=blue!14,line width=.9pt},
				unused/.style={tablevertex,draw=gray!55,fill=gray!5,text=gray!65},
				zeroarc/.style={-{Stealth[length=1.5mm]},blue!65!black,line width=.55pt},
				unusedarc/.style={-{Stealth[length=1.4mm]},gray!42,line width=.45pt},
				onearc/.style={-{Stealth[length=1.5mm]},ForestGreen,line width=.6pt},
				every node/.style={font=\small}
			]
			\node[treevertex,label=above:{$a=\mathrm{LCA}_i(u,v)$}] (a) at (1.9,2.55) {};
			\node[treevertex] (ul) at (.95,1.45) {};
			\node[treevertex] (vl) at (2.85,1.45) {};
			\node[treevertex,label=below:$u$] (u) at (.15,.35) {};
			\node[treevertex,label=below:$v$] (v) at (3.65,.35) {};
			\draw[treeedge] (a)--(ul)--(u);
			\draw[treeedge] (a)--(vl)--(v);
			\draw[originaledge] (u) to[bend right=28] node[below=2pt] {$e=(u,v)$} (v);
			\node[font=\scriptsize,rotate=49,fill=white,inner sep=.5pt] at (.78,1.38) {$u\leadsto a$};
			\node[font=\scriptsize,rotate=-49,fill=white,inner sep=.5pt] at (3.02,1.38) {$v\leadsto a$};
			\node[font=\footnotesize] at (1.9,-1.32) {(a) LCA decomposition};

			\begin{scope}[xshift=.35cm]
			\node[font=\scriptsize,text=gray!75!black,anchor=east] at (4.42,2.38) {$j=2$};
			\node[font=\scriptsize,text=gray!75!black,anchor=east] at (4.42,1.43) {$j=1$};
			\node[font=\scriptsize,text=gray!75!black,anchor=east] at (4.42,.48) {$j=0$};

			\node[selected] (x02) at (4.85,2.38) {$x_{0,2}$};
			\node[unused]   (x12) at (5.85,2.38) {$x_{1,2}$};
			\node[selected] (x22) at (6.85,2.38) {$x_{2,2}$};
			\node[unused]   (x32) at (7.85,2.38) {$x_{3,2}$};
			\node[unused]   (x42) at (8.85,2.38) {$x_{4,2}$};
			\node[unused]   (x52) at (9.85,2.38) {$x_{5,2}$};

			\node[tablevertex] (x01) at (4.85,1.43) {$x_{0,1}$};
			\node[unused]      (x11) at (5.85,1.43) {$x_{1,1}$};
			\node[tablevertex] (x21) at (6.85,1.43) {$x_{2,1}$};
			\node[unused]      (x31) at (7.85,1.43) {$x_{3,1}$};
			\node[tablevertex] (x41) at (8.85,1.43) {$x_{4,1}$};
			\node[unused]      (x51) at (9.85,1.43) {$x_{5,1}$};

			\foreach \r/\x in {0/4.85,1/5.85,2/6.85,3/7.85,4/8.85,5/9.85}{
				\node[tablevertex] (x\r0) at (\x,.48) {$x_{\r,0}$};
			}

			\node[edgevertex] (e) at (5.85,3.20) {$e$};
			\draw[zeroarc] (e)--node[pos=.48,left,yshift=3pt,inner sep=.5pt] {\scriptsize $0$} (x02);
			\draw[zeroarc] (e)--(x22);

			\draw[unusedarc] (x12)--(x11);
			\draw[unusedarc] (x12)--(x31);
			\draw[unusedarc] (x11)--(x10);
			\draw[unusedarc] (x11)--(x20);
			\draw[unusedarc] (x31)--(x30);
			\draw[unusedarc] (x31)--(x40);
			\draw[unusedarc] (x32)--(x31);
			\draw[unusedarc] (x32)--(x51);
			\draw[unusedarc] (x42)--(x41);
			\draw[unusedarc] (x42)--(10.85,1.43);
			\draw[unusedarc] (x52)--(x51);
			\draw[unusedarc] (x52)--(11.85,1.43);
			\draw[unusedarc] (x51)--(x50);
			\draw[unusedarc] (x51)--(10.85,.48);

			\draw[zeroarc] (x02)--(x01);
			\draw[zeroarc] (x02)--(x21);
			\draw[zeroarc] (x22)--(x21);
			\draw[zeroarc] (x22)--(x41);
			\draw[zeroarc] (x01)--(x00);
			\draw[zeroarc] (x01)--(x10);
			\draw[zeroarc] (x21)--(x20);
			\draw[zeroarc] (x21)--(x30);
			\draw[zeroarc] (x41)--(x40);
			\draw[zeroarc] (x41)--(x50);

			\foreach \j/\x in {0/4.35,1/5.35,2/6.35,3/7.35,4/8.35,5/9.35,6/10.35}{
				\node[treevertex] (p\j) at (\x,-.38) {};
			}
			\draw[treeedge] (p0)--(p1)--(p2)--(p3)--(p4)--(p5)--(p6);
			\foreach \h/\x in {1/4.85,2/5.85,3/6.85,4/7.85,5/8.85,6/9.85}{
				\node[edgevertex] (f\h) at (\x,-.38) {$f_{\h}$};
			}
			\foreach \r/\h in {0/1,1/2,2/3,3/4,4/5,5/6}{
				\draw[onearc] (x\r0)--(f\h);
			}
			\node[font=\scriptsize,ForestGreen,fill=white,inner sep=.5pt] at (4.62,.03) {$1$};
			\node[font=\scriptsize,anchor=north] at (4.35,-.56) {$u=x_0$};
			\node[font=\scriptsize,anchor=north] at (10.35,-.56) {$a=x_6$};
			\draw[gray!65,dashed,-{Stealth[length=1.5mm]}] (p6)--(11.15,-.38);
			\node[font=\tiny,text=gray!75!black,anchor=south] at (10.75,-.33) {rootward};
			\node[font=\scriptsize,text=gray!75!black,anchor=west] at (11.22,-.38) {$\mathrm{root}_i$};
			\node[font=\footnotesize,align=center] at (8.3,-1.32)
			{(b) Auxiliary connections and vertices reachable from $e$\\[-1pt]
				when the $u$--$a$ path has six edges};
			\end{scope}
		\end{tikzpicture}
		\caption{Compression of a tree exchange arc. (a) The path $P_{T_i}(e)$ is split at $a=\mathrm{LCA}_i(u,v)$. (b) For a six-edge $u$--$a$ branch, write $x_r=p_i^r(u)$ and $f_{r+1}=\mathrm{edge}_i(x_r)$. The blocks rooted at $x_{0,2}$ and $x_{2,2}$ cover $f_1{:}f_4$ and $f_3{:}f_6$ and share $x_{2,1}$. Gray nodes are unselected table entries. Blue and green arcs have weights $0$ and $1$, respectively, so $e$ reaches $f_1,\ldots,f_6$.}
		\label{fig:binary-lifting-compression}
	\end{figure}
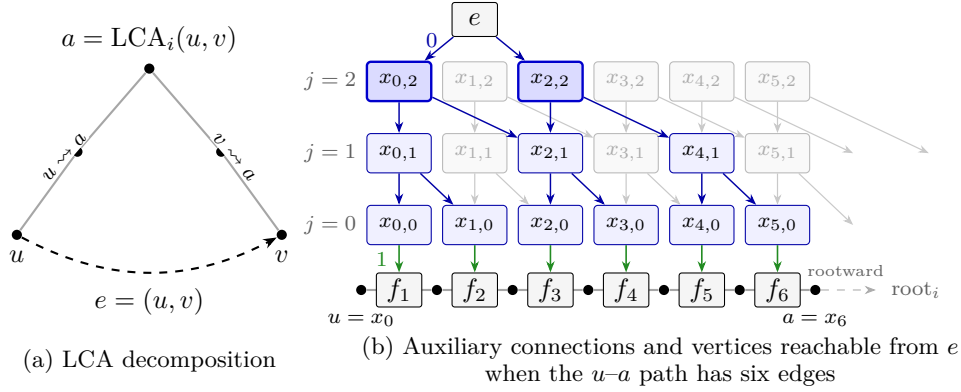

	For the same-head exchange arcs, create one additional virtual vertex $v'$ for each $v$ and add
	$$
		\begin{cases}
		(e,v') & \text{with weight 0 for each $e\in \rho_T(v)$}              \\
		(v',e) & \text{with weight 1 for each $e\in \rho_{E\setminus T}(v)$}
		\end{cases}
	$$
	These arcs ensure that every $e\in\rho_T(v)$ reaches all edge-vertices in $\rho_{E\setminus T}(v)$ at cost $1$.

	Finally, retain every deficit-source and joining-sink arc with weight $1$. There are $O(nk\log n)$ block vertices and block arcs, at most four connector arcs for each pair $(e,i)$, and $O(m)$ arcs in the same-head and retained gadgets. Hence the claimed size bounds hold.

	Every arc of $D(T)$ is now either retained with weight $1$ or can be represented by a directed segment of total weight $1$ whose internal vertices are virtual. Conversely, every maximal segment of a path in $D'(T)$ whose internal vertices are virtual contracts to exactly one valid arc of $D(T)$. Expanding or contracting these segments proves both path correspondences.
\end{proof}

Since $D'(T)$ is 0/1-weighted, we use 0--1 BFS to find shortest paths.

\begin{lemma}
	A shortest path in a 0/1-weighted digraph can be found in time linear in its number of arcs.
\end{lemma}
\begin{proof}[Proof Sketch]
	Use a deque to store the explored vertices. When traversing a weight-$0$ arc, add the new vertex to the front; when traversing a weight-$1$ arc, add it to the back.
\end{proof}

Finally, the digraph $D'(T)$ can be built and updated efficiently as follows.

\begin{lemma}
	The digraph $D'(T)$ can be built in $O(nk\log n+mk)$ time and updated in $O(k)$ time after each edge insertion.
\end{lemma}
\begin{proof}[Proof Sketch]
	First build the arcs between virtual vertices in $O(nk\log n)$ time.

	Next initialize, for each tree $T_i$ with $i<k$ and for all $T_k$-components, a data structure with linear initialization time and $O(1)$-time LCA queries. This takes $O(nk)$ time. For each original or inserted edge, spend $O(k)$ time querying the LCA and endpoint depths in every $T_i$ and adding the arcs to the appropriate virtual vertices. This takes $O(mk)$ total time.
\end{proof}

Therefore, replacing $D(T)$ by $D'(T)$ in \cref{alg:brute_reach} and \cref{alg:brute_resume_reach} reduces the brute-force search time to $O(nk\log n+mk)$. Plugging this bound into \cref{alg:round_robin} gives $O(mk^2\log^2 n)$ total update time.

\section{Data Structure for Scanning $P_{T_i}(e)$}
\label{sec:scan_C}

We describe the data structure used to efficiently scan through the unexplored edges in $P_{T_i}(e)$ in \cref{line:scan_C} of \cref{alg:reach}. By invariant~(1) in \cref{lem:cyc_scan_invariant}, $L_i$ is a connected subtree of $T_i$, and $e$ has exactly one endpoint in $L_i$.

\begin{theorem}[Data Structure for Scanning $P_{T_i}(e)$ \cite{GABOW1995259}]
	\label{thm:scan_C}
	Let $T$ be an undirected tree. There exists a data structure that takes $O(|V(T)|)$ space and maintains a connected subtree $L \subseteq V(T)$, where $T[L]$ contains no $T_k$-joining edge, that supports:
	\begin{itemize}
		\item $\textsc{Init}()$: Choose an arbitrary vertex $r$ in $T$ as root, and compute the depth and parent of each vertex in $O(|V(T)|)$ time.
		\item $\textsc{Reset}(s)$: Reset $L \gets \{s\}$ in $O(|L_{old}|)$ time.
		\item $\textsc{Query}(u, v)$: Provided that $|\{u,v\} \cap L|=1$, let $R:=V(P_T(u,v))\setminus L$.
		      If $P_T(u,v)$ contains no $T_k$-joining edge, the data structure updates $L\gets L\cup V(P_T(u,v))$ and returns the vertices in $R$ ordered by their distance from $u$. This operation runs in $O(|R|)$ time.

		      If $P_T(u,v)$ contains a $T_k$-joining edge, the procedure terminates and reports such an edge immediately upon encountering it. Its running time is linear in the number of edges traversed before the edge is encountered.
	\end{itemize}
\end{theorem}

\begin{remark}
If an additional $O(\log n)$ factor is acceptable, these operations can alternatively be implemented with a standard heavy--light decomposition: on each heavy chain, maintain the contiguous segment of explored edges and scan only the uncovered portions of the query path. A standard link--cut tree can also be used. Assuming $u\in L$, $\textsc{MakeRoot}(u)$ followed by $\textsc{Access}(v)$ exposes $P_T(u,v)$ as an auxiliary splay tree; we then traverse its suffix outside $L$, splaying each visited vertex.
\end{remark}

\begin{proof}

	During initialization, the algorithm roots $T$ at $r$ and computes each vertex's depth and parent.

	The algorithm explicitly maintains $r'\in L$ as the shallowest vertex in $L$. On reset, it sets $L\gets\{s\}$ and $r'\gets s$.

	For a query $(u,v)$, without loss of generality, assume $u\in L$. Since $L$ is connected, the path from $u$ to $r'$ is contained in $L$. The algorithm starts with two pointers at $r'$ and $v$, and repeatedly moves the deeper pointer to its parent. It stops when the pointer from $v$ is about to enter $L$, when the next edge is a $T_k$-joining edge, or when the two pointers meet.

	If a $T_k$-joining edge is encountered, the algorithm reports it immediately. Otherwise, it records the vertices traversed by the two pointers, reverses the order of the vertices recorded from the $v$-side, concatenates the two lists, and returns them. The vertex $r'$ is unchanged if the pointer from $v$ enters $L$, and is otherwise updated to the meeting vertex.

	Since every traversed vertex is either returned and added to $L$, or precedes the reported joining edge, the running time is linear in the number of traversed vertices.
\end{proof}

Thus, it suffices to maintain the above data structure for each $T_i (i \in [k-1])$, and for each connected $T_k$-component in $T_k$, so that we can answer the query on \cref{line:scan_C} of \cref{alg:reach} by calling $\textsc{Query}(u,v),$ where $e=(u,v)$.

\end{document}